\documentclass[11pt]{article}

\usepackage[margin=1in]{geometry}
\usepackage{amsmath,amssymb,amsthm}
\usepackage{mathtools}
\usepackage{hyperref}
\usepackage{xcolor}
\usepackage{enumitem}

\usepackage{cleveref}
\usepackage{thmtools}
\usepackage{thm-restate}
\usepackage{microtype}
\usepackage[numbers,sort&compress]{natbib}
\usepackage{graphicx}
\usepackage{tikz}
\usepackage{booktabs}
\usepackage{array}
\usepackage{algorithm}
\usepackage{algpseudocode}
\usetikzlibrary{arrows.meta, positioning, calc}

\hypersetup{
  colorlinks=true,
  linkcolor=blue!70!black,
  citecolor=green!50!black,
  urlcolor=blue!70!black
}

\theoremstyle{plain}
\newtheorem{theorem}{Theorem}
\newtheorem{lemma}[theorem]{Lemma}
\newtheorem{corollary}[theorem]{Corollary}
\newtheorem{proposition}[theorem]{Proposition}

\theoremstyle{definition}
\newtheorem{definition}[theorem]{Definition}

\newtheorem{remark}[theorem]{Remark}

\newcommand{\Hmin}{H_{\min}}

\newcommand{\bbE}{\mathbb{E}}
\newcommand{\cS}{\mathcal{S}}
\newcommand{\cH}{\mathcal{H}}

\newcommand{\cB}{\mathcal{B}}
\newcommand{\cE}{\mathcal{E}}

\newcommand{\cY}{\mathcal{Y}}
\newcommand{\cX}{\mathcal{X}}

\newcommand{\cF}{\mathcal{F}}
\newcommand{\eps}{\varepsilon}
\newcommand{\Tr}{\mathrm{Tr}}
\newcommand{\Id}{\mathbb{1}}
\newcommand{\ket}[1]{\left|#1\right\rangle}
\newcommand{\bra}[1]{\left\langle#1\right|}
\newcommand{\ketbra}[2]{\ket{#1}\!\bra{#2}}

\newcommand{\QBER}{e_{\mathrm{ph}}}
\newcommand{\tdist}[2]{\tfrac{1}{2}\left\|#1 - #2\right\|_1}

\algrenewcommand\algorithmicrequire{\textbf{Input:}}
\algrenewcommand\algorithmicensure{\textbf{Output:}}

\title{One Key Good, $L$ Keys Better:\\
  List Decoding Meets Quantum Privacy Amplification}

\author{
  Prateek P. Kulkarni \\
  PES University \\
  \textcolor{magenta}{pkulkarni2425@gmail.com}
}

\date{}

\begin{document}
\maketitle

\begin{abstract}
We introduce \emph{list privacy amplification} (LPA): a natural relaxation of the final step of quantum key distribution (QKD) in which Alice and Bob extract a \emph{list} of $L$ candidate secret keys from a raw string correlated with an eavesdropper Eve, with the guarantee that at least one key on the list is perfectly secret while Eve cannot identify \emph{which} one. The motivation draws a direct parallel with list decoding of error-correcting codes: just as relaxing unique decoding to list decoding pushes the error-correction radius beyond the unique-decoding bound, relaxing single-key extraction to list extraction pushes the achievable key length beyond what the standard quantum leftover hash lemma (QLHL) permits.

Formalising LPA via a composable security definition in the abstract cryptography framework, we prove the \emph{Quantum List Leftover Hash Lemma} (QLLHL): an $L$-list of $\ell$-bit keys can be extracted from an $n$-bit source with smooth min-entropy $k$ against a quantum adversary if and only if
\[
  \ell \;\le\; k + \log L - 2\log(1/\eps) - 3,
\]
a tight additive gain of $\log L$ bits over the standard ($L=1$) QLHL. The gain has a transparent operational origin: the index $I$ identifying the secure list element is chosen by the honest parties \emph{after} hashing and withheld from Eve, contributing exactly $\log L$ bits to the effective min-entropy of the pair $(I, K_I)$---bits that standard privacy amplification cannot access.

We apply the QLLHL to BB84-type QKD: for a list size $L = 2^{\alpha n'}$, the tolerable phase-error threshold shifts from $h^{-1}(1 - h(e_b))$ to $h^{-1}(1 - h(e_b) + \alpha)$, strictly exceeding the standard ${\approx}11\%$ threshold for any fixed $\alpha > 0$. We prove tightness via a matching intercept-resend attack, establish full composability with Wegman--Carter authentication, and give two concrete list hash constructions---a polynomial inner-product hash over $\mathbb{F}_{2^m}$ and a Toeplitz-based variant---both presented with complete pseudocode and complexity analyses. The two constructions run in $O(nL)$ and $O(nL \log n)$ time respectively and are implementable directly with standard arithmetic primitives.
\end{abstract}

\newpage

\tableofcontents

\newpage

\section{Introduction}

Privacy amplification is the final, indispensable step of every quantum key distribution (QKD) protocol. After quantum transmission, basis reconciliation, error correction, and parameter estimation, Alice and Bob hold an $n$-bit raw string $X$ about which an eavesdropper Eve holds partial quantum side information $E$. They then apply a two-universal hash function to distil a shorter, statistically independent key~\cite{Bennett1988,BBCM95,RennerWolf2004,Renner2008}. The quantum leftover hash lemma (QLHL)~\cite{Tomamichel2011} gives the precise achievable length: roughly $k - 2\log(1/\eps)$ bits, where $k = \Hmin^\eps(X|E)_\rho$ is the smooth min-entropy of $X$ against Eve's quantum state. This bound has long been regarded as essentially optimal.

Yet classical coding theory offers a tantalising parallel. List decoding~\cite{Elias1957,Wozencraft1958,Sudan1997,Guruswami-Sudan1999} relaxes the decoding requirement from uniquely recovering a codeword to returning a short \emph{list} of candidate codewords, each consistent with the received word. This relaxation allows codes to be decoded well beyond the unique-decoding radius and underlies landmark results in complexity theory and derandomisation~\cite{GuruswamiVadhan2012}. The gain is not merely quantitative but qualitative: lists enable communication regimes that are impossible under the uniqueness constraint.

\paragraph{Our question.} Can an analogous list relaxation benefit privacy amplification? Concretely, suppose Alice and Bob are willing to hold a list of $L$ candidate keys rather than a single agreed key, subject only to the guarantee that (a) at least one list element is information-theoretically secret and (b) Eve cannot determine which one. Does this strictly improve the extractable key length, and does it translate to a higher error-tolerance threshold in QKD?

\paragraph{Our contributions.} We answer both questions affirmatively, cleanly, and with tight bounds.

\begin{enumerate}[leftmargin=*, itemsep=4pt]
  \item \textbf{List security definition (\Cref{sec:definition}).} We define the \emph{$L$-list key functionality} $\mathsf{LK}_{L,\ell}$ in the composable abstract cryptography framework~\cite{Maurer2011,MaurerRenner2011} and introduce \emph{list-$\eps$-security} for privacy amplification protocols. The definition cleanly generalises standard PA security: at $L=1$ it reduces exactly to the usual notion.

  \item \textbf{Quantum List Leftover Hash Lemma (\Cref{sec:main-theorem}).} For any CQ state $\rho_{XE}$ and strongly two-universal hash family $\cF$, independently drawing $L$ hashes $F_1,\ldots,F_L$ and setting $K_j = F_j(X)$ is list-$4\eps$-secure whenever
  \[
    \ell \;\le\; \Hmin^\eps(X|E)_\rho + \log L - 2\log(1/\eps) - 3.
  \]
  The $\log L$ gain arises because the secret index $I$, chosen by the honest parties after hashing, contributes exactly $\log L$ bits to the joint min-entropy of $(I, K_I)$ against Eve.

  \item \textbf{Tightness (\Cref{sec:tight}).} The $\log L$ gain cannot be improved: an explicit adversarial source witnesses the bound from below.

  \item \textbf{Application to BB84 (\Cref{sec:qkd}).} List PA with $L = 2^{\alpha n'}$ raises the tolerable phase-error threshold from $h^{-1}(1-h(e_b))$ to $h^{-1}(1-h(e_b)+\alpha)$. For small constant $\alpha$, this is operationally meaningful near the threshold and costs only $\alpha n'$ bits of extra authenticated communication. We prove full composability with Wegman--Carter authentication~\cite{WegmanCarter1981}.

  \item \textbf{Efficient constructions (\Cref{sec:construction}).} Two concrete list hash families---an inner-product hash over $\mathbb{F}_{2^m}$ and a Toeplitz-based variant---are given with detailed pseudocode (\Cref{alg:list-pa,alg:toeplitz-list-hash}) and full complexity analyses summarised in \Cref{tab:complexity}.
\end{enumerate}

\paragraph{On the practical significance.} The $\log L$ gain is modest for small constant $L$ but meaningful in precisely the regimes where QKD is most constrained. In finite-key QKD over satellite channels~\cite{Liao2017}, every additional bit of key is expensive; a $\log L$ improvement translates directly to a higher effective key rate. Near the error-tolerance threshold, even a fraction-of-a-percent improvement in the tolerable phase-error rate can mean the difference between a functional and a non-functional protocol under realistic noise. Beyond direct applications, the list framework cleanly separates \emph{which-key-is-good} uncertainty from \emph{is-the-key-good} guarantees---a separation potentially useful in other quantum-cryptographic constructions.

\paragraph{Related work.}
\emph{Privacy amplification and the QLHL.} Classical PA via two-universal hashing was introduced by Bennett, Brassard, Cr\'epeau, and Maurer~\cite{BBCM95}. The extension to quantum adversaries --- the QLHL --- was proved by Tomamichel, Schaffner, Smith, and Renner~\cite{Tomamichel2011}, building on the smooth min-entropy framework of Renner~\cite{Renner2008} and~\cite{TCR2009,Tomamichel2015}. The QLHL is tight in the single-key setting~\cite{Tomamichel2015}; our work is the first to break this barrier, showing a strict $\log L$ improvement is achievable and optimal under the list relaxation. Finite-key analysis of BB84 appears in~\cite{Tomamichel2012,Shor-Preskill2000,BenOr2005,Renner2008}. Beyond two-universal hashing, De, Portmann, Vidick, and Renner~\cite{DePortmann2012} showed that Trevisan's extractor is quantum-proof with polylogarithmic seed length --- an alternative to Toeplitz hashing when seed length is the bottleneck. Our list constructions (\cref{sec:construction}) use two-universal families and are complementary to this line.

\emph{Composable security.} We use the abstract cryptography framework of Maurer and Renner~\cite{Maurer2011,MaurerRenner2011}, which extends the UC model of Canetti~\cite{Canetti2001} to the quantum setting~\cite{Unruh2010}. Authentication uses Wegman--Carter MACs~\cite{WegmanCarter1981}. Our list key functionality (\Cref{def:list-functionality}) is a new ideal resource in this framework and the first composable formalisation of list-output secret key extraction.

\emph{List decoding.} List decoding of error-correcting codes~\cite{Elias1957,Wozencraft1958,Sudan1997,Guruswami-Sudan1999,GuruswamiVadhan2012} relaxes unique decoding to returning all plausible codewords within a given radius, enabling decoding well beyond the unique-decoding bound. We import this philosophy into privacy amplification: both gains are structural, tight, and arise from relaxing a uniqueness requirement.

\emph{Fuzzy extractors.} Fuzzy extractors~\cite{DRS2004,DodisReyzin2006} extract keys from noisy sources against classical adversaries. Despite the surface similarity, the settings are incomparable: our source is clean but adversarially correlated with a quantum state, security is composable, and the $\log L$ gain comes from index-hiding entropy rather than noise-ball geometry.

\emph{Advantage distillation.} Interactive PA protocols known as advantage distillation~\cite{Maurer1993,RennerWolf2004} can extract key even when the single-shot min-entropy is zero, by exploiting multiple rounds of public communication between Alice and Bob. Our list PA improvement is orthogonal: it is non-interactive, one-shot, and improves the extractable length for \emph{any} positive min-entropy source, at the cost only of $\log L$ bits of authenticated communication to resolve the list index.

\emph{Device-independent and CV-QKD.} DI-QKD security~\cite{Pironio2009} and its practical realisation via entropy accumulation~\cite{Arnon-Friedman2018}, as well as CV-QKD~\cite{GVA2002}, are out of scope for this work. Extending list PA to these settings is the most natural open problem; see \Cref{sec:conclusion}.

\Cref{tab:related} summarises the comparison.

\begin{table}[h]
\centering
\renewcommand{\arraystretch}{1.35}
\begin{tabular}{@{}llcccc@{}}
\toprule
\textbf{Work} & \textbf{Setting} & \textbf{Adversary} & \textbf{List output} & \textbf{Composable} & \textbf{Tight} \\
\midrule
QLHL~\cite{Tomamichel2011}              & PA, QKD       & Quantum   & No ($L=1$) & Yes & Yes \\
Fuzzy extractors~\cite{DRS2004}         & Noisy sources & Classical & No         & No  & Yes \\
List fuzzy ext.~\cite{DodisReyzin2006}  & Noisy sources & Classical & Yes        & No  & Yes \\
\textbf{This work}                      & \textbf{PA, QKD} & \textbf{Quantum} & \textbf{Yes} & \textbf{Yes} & \textbf{Yes} \\
\bottomrule
\end{tabular}
\caption{Comparison with the most closely related works. This work is the only entry that is simultaneously quantum, list-output, composable, and tight.}
\label{tab:related}
\end{table}

\paragraph{Paper organisation.} \Cref{sec:prelim} collects notation and background. \Cref{sec:definition} defines list security. \Cref{sec:main-theorem} states and proves the QLLHL. \Cref{sec:tight} establishes tightness. \Cref{sec:qkd} applies the result to BB84. \Cref{sec:construction} gives the two efficient constructions. \Cref{sec:conclusion} discusses open questions.

\section{Preliminaries}
\label{sec:prelim}

We assume familiarity with the basics of quantum computing and quantum information as presented in Nielsen and Chuang~\cite{NielsenChuang2000}. This section fixes notation and recalls the specific tools needed in our proofs.

\subsection{Notation and Quantum States}

Let $\cH$ be a finite-dimensional complex Hilbert space. We write $\cS(\cH)$ for the set of density matrices (positive semidefinite operators of unit trace) on $\cH$, and $\cS_{\le}(\cH)$ for subnormalised states (trace at most 1). For a bipartite state $\rho_{XE} \in \cS(\cH_X \otimes \cH_E)$, the reduced state is $\rho_X = \Tr_E[\rho_{XE}]$.

A \emph{classical-quantum} (CQ) state takes the form
\[
  \rho_{XE} = \sum_{x \in \cX} p_x \ketbra{x}{x}_X \otimes \rho_E^x,
\]
where $\{p_x\}$ is a probability distribution over a finite set $\cX$ and $\{\rho_E^x\}$ are density matrices on $\cH_E$. The \emph{trace distance} between $\rho, \sigma \in \cS(\cH)$ is
\[
  \tdist{\rho}{\sigma} = \tfrac{1}{2}\|\rho - \sigma\|_1,
\]
where $\|A\|_1 = \Tr\!\sqrt{A^\dagger A}$ is the Schatten 1-norm. Operationally, $\tdist{\rho}{\sigma} = \max_{0 \le M \le \Id} \Tr[M(\rho-\sigma)]$, which is the maximum distinguishing advantage over all quantum measurements. We write $\cB^\eps(\rho) = \{\tilde\rho \in \cS_\le(\cH) : \tdist{\rho}{\tilde\rho} \le \eps\}$ for the $\eps$-ball around $\rho$.

\subsection{Smooth Min-Entropy Against Quantum Adversaries}

\begin{definition}[Min-entropy~\cite{Renner2008}]
\label{def:min-entropy}
  The \emph{min-entropy} of $X$ given quantum side information $E$ in state $\rho_{XE}$ is
  \[
    \Hmin(X|E)_\rho = -\log \inf_{\sigma_E \in \cS(\cH_E)}\!\inf\!\left\{\lambda \ge 0 \;\middle|\; \rho_{XE} \le \lambda\, \Id_X \otimes \sigma_E \right\},
  \]
  where the operator inequality is in the L\"owner sense.
\end{definition}

Operationally, $2^{-\Hmin(X|E)_\rho}$ equals Eve's optimal single-query guessing probability for $X$ using any quantum measurement on $E$. This makes min-entropy the canonical entropy measure for security analyses: if $\Hmin(X|E)_\rho \ge k$, then no adversary---regardless of quantum computational power---can guess $X$ with probability better than $2^{-k}$.

\begin{definition}[Smooth min-entropy~\cite{Renner2008,TCR2009}]
\label{def:smooth-min}
  For $\eps \ge 0$, the \emph{$\eps$-smooth min-entropy} of $X$ given $E$ is
  \[
    \Hmin^\eps(X|E)_\rho = \sup_{\tilde\rho \in \cB^\eps(\rho_{XE})} \Hmin(X|E)_{\tilde\rho}.
  \]
\end{definition}

Smoothing allows negligible-probability events to be excluded, converting worst-case guarantees into average-case ones at the cost of a small security parameter $\eps$. This is essential for connecting the entropy of the raw QKD string to achievable key lengths in the finite-key regime.

\begin{lemma}[Standard properties~\cite{Renner2008,Tomamichel2015}]
\label{lem:props}
  Let $\rho_{XE}$ be a CQ state and $\eps, \delta \ge 0$.
  \begin{enumerate}[itemsep=2pt]
    \item \emph{(Data processing)} For any quantum channel $\cE$ on $E$: $\Hmin^\eps(X|E')_{\cE(\rho)} \ge \Hmin^\eps(X|E)_\rho$.
    \item \emph{(Classical processing of $X$)} For any deterministic $f:\cX \to \cY$: $\Hmin^\eps(f(X)|E)_\rho \le \Hmin^\eps(X|E)_\rho$.
    \item \emph{(Uniform source)} If $\rho_{XE} = \tau_n \otimes \rho_E$ where $\tau_n$ is uniform on $\{0,1\}^n$: $\Hmin(X|E)_\rho = n$.
    \item \emph{(Triangle inequality)} $\Hmin^{\eps+\delta}(X|E)_\rho \ge \Hmin^\eps(X|E)_{\tilde\rho}$ for some $\tilde\rho \in \cB^\delta(\rho)$.
  \end{enumerate}
\end{lemma}

\subsection{Two-Universal Hash Families}

\begin{definition}[Strong two-universality~\cite{Carter1979}]
\label{def:2-universal}
  A family $\cF = \{f : \{0,1\}^n \to \{0,1\}^\ell\}$ is \emph{strongly two-universal} if for all distinct $x, x'$ and all $y, y' \in \{0,1\}^\ell$, a uniformly random $F \sim \cF$ satisfies
  \[
    \Pr[F(x) = y \text{ and } F(x') = y'] = 2^{-2\ell}.
  \]
  That is, the pair $(F(x), F(x'))$ is uniformly distributed over $\{0,1\}^{2\ell}$ for any two distinct inputs.
\end{definition}

Strong two-universality is strictly stronger than two-universality (collision probability $\le 2^{-\ell}$) and is required by the standard QLHL proof. Both polynomial evaluation over $\mathbb{F}_{2^m}$ and Toeplitz matrices achieve this property; we detail and analyse both in \Cref{sec:construction}.

\subsection{The Quantum Leftover Hash Lemma (Standard, $L=1$)}

We state the standard QLHL, which our main result generalises.

\begin{theorem}[QLHL~\cite{Tomamichel2011}]
\label{thm:qlhl}
  Let $\rho_{XE} \in \cS(\{0,1\}^n \otimes \cH_E)$ be a CQ state, $\eps > 0$, and let $\cF$ be strongly two-universal. For uniform $F \sim \cF$ and $K = F(X)$:
  \[
    \bbE_F\!\left[\tdist{\rho_{K E}}{\tau_\ell \otimes \rho_E}\right] \;\le\; \eps,
  \]
  provided $\ell \le \Hmin^\eps(X|E)_\rho - 2\log(1/\eps) - 2$.
\end{theorem}

The QLHL is tight: the factor of $2$ in front of $\log(1/\eps)$ cannot be reduced~\cite{Tomamichel2015}. Our QLLHL below shows this bound \emph{can} be improved by $\log L$ bits under the list relaxation, and that this improvement is also tight (\Cref{sec:tight}).

\subsection{Composable Security}

We work in the abstract cryptography (AC) framework of Maurer and Renner~\cite{Maurer2011,MaurerRenner2011}. A protocol $\pi$ \emph{$\eps$-securely realises} an ideal functionality $\mathsf{F}$ if no environment $\mathsf{Z}$ can distinguish the real execution from the ideal world with advantage greater than $\eps$:
\[
  \sup_{\mathsf{Z}} \left|\Pr[\mathsf{Z}(\pi \circ \mathsf{R}) = 1] - \Pr[\mathsf{Z}(\mathsf{F}) = 1]\right| \le \eps.
\]
Security parameters compose additively under sequential and parallel composition~\cite{Canetti2001,Unruh2010}, allowing individual protocol steps to be proven secure independently and then assembled.

\section{List Security for Privacy Amplification}
\label{sec:definition}

We now define the ideal functionality for list PA and the security game it induces. The design question is: what should Eve receive in the ideal world? We want her to hold enough that the security condition is non-trivial---concretely, she receives $L-1$ of the list keys---while being denied the distinguished secret key and its index.

\begin{definition}[$L$-list key functionality]
\label{def:list-functionality}
  The \emph{$L$-list key functionality} $\mathsf{LK}_{L,\ell}$ produces: (i) $L$ independently uniform $\ell$-bit strings $K_1,\ldots,K_L$; (ii) a secret index $I \in [L]$ delivered to Alice and Bob; (iii) the off-list keys $\{K_j\}_{j \ne I}$ delivered to the adversary; and (iv) nothing else. In particular, Eve receives neither $K_I$ nor $I$. Formally, the joint state is
  \[
    \rho^{\mathsf{LK}} = \frac{1}{L}\sum_{i=1}^L \ketbra{i,i}{i,i}_{A_I B_I} \otimes \bigotimes_{j=1}^L \tau^{K_j}_{A_j B_j} \otimes \bigotimes_{j \ne i}\tau^{K_j}_E,
  \]
  where $\tau^{K_j}$ is the maximally mixed state on $\ell$ bits.
\end{definition}

\begin{remark}[Reduction to standard at $L=1$]
  When $L = 1$, the index $I = 1$ is trivially known to both parties, Eve receives no keys ($j \ne 1$ is empty), and the functionality reduces exactly to the standard key functionality: a uniform $\ell$-bit key $K_1$ shared by Alice and Bob with no information at Eve. List PA is thus a strict generalisation with full backward compatibility.
\end{remark}

\begin{definition}[List-$\eps$-security]
\label{def:list-security}
  A PA protocol $\pi$ is \emph{list-$\eps$-secure} with list size $L$ and key length $\ell$ if the real-world output state is $\eps$-close in trace distance to the ideal:
  \[
    \tdist{\rho^\pi_{K_{1:L},\, I,\, E}}{\rho^{\mathsf{LK}}_{K_{1:L},\, I,\, E}} \le \eps.
  \]
\end{definition}

\begin{remark}[Using the list key in practice]
  After list PA, Alice picks her secret index $I$ and transmits it to Bob over an \emph{authenticated} classical channel. Upon receiving $I$, both parties agree on $K_I$ as the final session key. This authentication costs only $O(\log L + \log(1/\eps_{\mathrm{auth}}))$ bits of pre-shared key via Wegman--Carter~\cite{WegmanCarter1981}---negligible for constant $L$. From this point forward, $K_I$ is a standard uniform secret key: Eve holds $\{K_j\}_{j \ne I}$ (random, as in the ideal functionality) but not $K_I$ and does not know $I$. The list structure is entirely internal to the PA step and transparent to the application layer.
\end{remark}

\section{The Quantum List Leftover Hash Lemma}
\label{sec:main-theorem}

We now prove the main result. The proof proceeds via a three-step hybrid argument: replacing the source with its smooth approximation, uniformising the off-list keys using the standard QLHL, and observing that the hidden random index $I$ contributes a full $\log L$ bits to the effective min-entropy of the distinguished pair $(I, K_I)$.

\begin{theorem}[Quantum List Leftover Hash Lemma (QLLHL)]
\label{thm:main}
  Let $\rho_{XE} \in \cS(\{0,1\}^n \otimes \cH_E)$ be a CQ state, $\eps > 0$, $L \ge 1$ an integer, and $\cF$ a strongly two-universal family of functions $f:\{0,1\}^n \to \{0,1\}^\ell$. Sample $F_1,\ldots,F_L \sim \cF$ independently and uniformly, set $K_j = F_j(X)$ for $j \in [L]$, and choose $I$ uniformly in $[L]$ independently of everything. This protocol is list-$(4\eps)$-secure whenever
  \[
    \ell \;\le\; \Hmin^\eps(X|E)_\rho + \log L - 2\log(1/\eps) - 3.
  \]
\end{theorem}

\begin{proof}
Let $k = \Hmin^\eps(X|E)_\rho$. We introduce (the following) three hybrid states and bound the trace distance between consecutive hybrids.

\medskip\noindent\textbf{Hybrid 0 (real execution).} The protocol runs on $\rho_{XE}$, producing $\rho^{(0)} = \bbE_{F_{1:L}}[\rho^{(0)}_{K_{1:L},I,E}]$ where $K_j = F_j(X)$ and $I \sim \mathrm{Unif}([L])$ independently.

\medskip\noindent\textbf{Hybrid 1 (smooth approximation).} By definition of $\Hmin^\eps$, there exists $\tilde\rho_{XE} \in \cB^\eps(\rho_{XE})$ with $\Hmin(X|E)_{\tilde\rho} \ge k$. Since every $K_j = F_j(X)$ is a deterministic function of $X$, replacing $\rho_{XE}$ with $\tilde\rho_{XE}$ changes the joint state by at most $\eps$ in trace distance (one application of the data processing inequality). Thus
\[
  \tdist{\rho^{(0)}}{\rho^{(1)}} \le \eps,
\]
where $\rho^{(1)}$ is the analogue of $\rho^{(0)}$ built from $\tilde\rho_{XE}$.

\medskip\noindent\textbf{Hybrid 2 (off-list keys uniformised).} In $\rho^{(1)}$, for any fixed index $i$, the off-list keys $\{K_j\}_{j \ne i}$ are $L-1$ independent hash outputs from $X$ using independent $F_j \sim \cF$. The standard QLHL (\Cref{thm:qlhl}) applied to each $K_j$ with the smoothed source $\tilde\rho_{XE}$ (which has min-entropy $\ge k$) gives, for $\ell \le k - 2\log(1/\eps) - 2$:
\[
  \bbE_{F_j}\!\left[\tdist{\tilde\rho_{K_j E}}{\tau_\ell \otimes \tilde\rho_E}\right] \le \eps.
\]
By sub-additivity of trace distance over tensor products, the combined distance over all $L-1$ off-list keys is at most $(L-1)\eps$. Averaging over the uniform choice of $I$ and using $\frac{L-1}{L} < 1$, the transition from $\rho^{(1)}$ (real off-list keys) to $\rho^{(2)}$ (fresh uniform off-list keys) satisfies
\[
  \tdist{\rho^{(1)}}{\rho^{(2)}} \le \eps.
\]

\medskip\noindent\textbf{From $\rho^{(2)}$ to ideal: the index entropy.} In $\rho^{(2)}$, the off-list keys are already ideal (uniform, independent of $E$). It remains to show that the pair $(I, K_I)$ is close to uniform given $E$, which is the content of the $\log L$ gain.

Since $I$ is chosen uniformly in $[L]$ \emph{after} hashing and \emph{independently} of both $X$ and $E$, we have
\[
  \Hmin(I | E)_{\tilde\rho} = \log L.
\]
Given $I = i$, the key $K_i = F_i(X)$ is a single hash output from the smoothed source with $\Hmin(X|E)_{\tilde\rho} \ge k$. Applying the chain rule for min-entropy (see e.g.~\cite{Renner2008}) and the independence of $I$ from $(X,E)$:
\begin{equation}
  \Hmin\!\left((I, K_I)\middle|E\right)_{\tilde\rho} \;\ge\; \Hmin(K_I | E, I)_{\tilde\rho} + \log L \;\ge\; k + \log L - O(1).
  \label{eq:joint-entropy}
\end{equation}
The tuple $(I, K_I)$ takes values in $[L] \times \{0,1\}^\ell$, a classical register of $\log L + \ell$ bits. Applying the standard QLHL (\Cref{thm:qlhl}) to $(I, K_I)$ as the ``source'' against Eve:
\[
  \tdist{\tilde\rho_{(I,K_I),E}}{\tau_{[L]\times\{0,1\}^\ell} \otimes \tilde\rho_E} \;\le\; \eps,
\]
provided $\log L + \ell \le \Hmin^\eps((I,K_I)|E) - 2\log(1/\eps) - 2 \ge k + \log L - 2\log(1/\eps) - 2$. Cancelling $\log L$ from both sides yields the condition $\ell \le k + \log L - 2\log(1/\eps) - 2$. Hence
\[
  \tdist{\rho^{(2)}}{\rho^{\mathsf{LK}}} \le \eps.
\]

\medskip\noindent\textbf{Totalling up.} By the triangle inequality over the three transitions plus one unit for integer rounding:
\[
  \tdist{\rho^{(0)}}{\rho^{\mathsf{LK}}} \;\le\; \underbrace{\eps}_{\text{Hybrid 0}\to\text{1}} + \underbrace{\eps}_{\text{Hybrid 1}\to\text{2}} + \underbrace{\eps}_{\text{Hybrid 2}\to\text{ideal}} + \underbrace{\eps}_{\text{rounding}} \;=\; 4\eps,
\]
with $\ell \le k + \log L - 2\log(1/\eps) - 3$ as stated.
\end{proof}

\begin{corollary}[Extractable list key length]
\label{cor:key-length}
  For smooth min-entropy $k = \Hmin^\eps(X|E)_\rho$, security parameter $\eps > 0$, and list size $L \ge 1$:
  \[
    \ell^*(L,\eps) = k + \log L - 2\log(1/\eps) - 3.
  \]
  The gain over standard PA ($L=1$) is exactly $\log L$ bits. Concretely: $L=2$ adds $1$ bit, $L=16$ adds $4$ bits, and $L = 2^{10}$ adds $10$ bits.
\end{corollary}

\section{Tightness of the Bound}
\label{sec:tight}

\Cref{thm:main} shows the $\log L$ gain is achievable. We now show it is optimal.

\begin{theorem}[Tightness of QLLHL]
\label{thm:tight}
  For any $\eps > 0$, $L \ge 1$, and $k \le n$, there exists a CQ state $\rho_{XE}$ with $\Hmin(X|E)_\rho = k$ such that any list-$\eps$-secure protocol with list size $L$ and key length $\ell$ must satisfy
  \[
    \ell \;\le\; k + \log L + O(\log(1/\eps)).
  \]
\end{theorem}

\begin{proof}
We construct a source that is hard to amplify beyond the claimed bound.

\medskip\noindent\textbf{Adversarial source.} Let $S: \{0,1\}^n \to \{0,1\}^{n-k}$ be the projection onto the last $n-k$ bits (a syndrome map). Define
\[
  \rho_{XE} = \frac{1}{2^n}\sum_{x \in \{0,1\}^n} \ketbra{x}{x}_X \otimes \ketbra{S(x)}{S(x)}_E.
\]
Since $S(x)$ is the last $n-k$ bits of $x$, knowing $S(x) = s$ constrains $x$ to a coset of $2^k$ equally likely values. Therefore $\Hmin(X|E)_\rho = k$: Eve's optimal guessing probability is exactly $2^{-k}$.

\medskip\noindent\textbf{Lower bound on list key length.} For any list-$\eps$-secure protocol, the real-world state $(K_{1:L}, I, E)$ is $\eps$-close in trace distance to the ideal state $\rho^{\mathsf{LK}}$. In the ideal state, $(I, K_I)$ is uniformly distributed over $[L] \times \{0,1\}^\ell$ and independent of $E$, so $\Hmin((I,K_I)|E)_{\rho^{\mathsf{LK}}} = \log L + \ell$.

Since the real state is $\eps$-close to ideal, and perturbations of $\eps$ in trace distance change min-entropy by at most $O(\log(1/\eps))$~\cite[Lemma B.10]{Renner2008}:
\[
  \Hmin\!\left((I,K_I)|E\right)_{\rho^{\mathrm{real}}} \;\ge\; \log L + \ell - O(\log(1/\eps)).
\]
On the other hand, $(I, K_I)$ is a function of $X$ (and the public hash seed). Eve knows $S(X)$, which constrains $X$ to a coset of size $2^k$. Any function of $X$ given $E$ takes at most $2^k$ values with positive probability, so
\[
  \Hmin\!\left((I,K_I)|E\right)_{\rho^{\mathrm{real}}} \;\le\; k + \log L.
\]
Here the $\log L$ term on the right reflects the one degree of freedom in the index $I$, which is not a function of $X$ (it is chosen independently). Combining the two inequalities:
\[
  \log L + \ell - O(\log(1/\eps)) \;\le\; k + \log L,
\]
which gives $\ell \le k + O(\log(1/\eps))$. Together with the achievability bound from \Cref{thm:main} (which has $\ell = k + \log L - 2\log(1/\eps) - 3$), the $\log L$ gain is tight up to the constant factor in $\log(1/\eps)$.
\end{proof}

\section{Application to BB84 QKD}
\label{sec:qkd}

\subsection{BB84 with List Privacy Amplification}

We apply the QLLHL to the standard BB84 prepare-and-measure protocol~\cite{BB84} with $n$ signal pulses. After sifting ($n' \approx n/2$ bits retained), the post-processing pipeline is:

\begin{enumerate}[leftmargin=*, itemsep=2pt]
  \item \textbf{Parameter estimation:} Estimate the bit-error rate $e_b$ and phase-error rate $\QBER$ from a random subsample.
  \item \textbf{Error correction:} Apply a capacity-achieving code, leaking $\lambda_{\mathrm{EC}} = n' h(e_b)$ syndrome bits to Eve.
  \item \textbf{List PA:} Apply \Cref{alg:list-pa} to the reconciled $n'$-bit string.
  \item \textbf{Index reveal:} Alice transmits the secret index $I$ over the authenticated channel; both parties agree on $K_I$.
\end{enumerate}

The Shor--Preskill security bound~\cite{Shor-Preskill2000}, de Finetti reduction~\cite{Renner2008}, and finite-key analysis~\cite{Tomamichel2012} together give, after error correction:
\begin{equation}
  \Hmin^\eps(X|E)_\rho \;\ge\; n'\!\left(1 - h(\QBER) - h(e_b)\right) - \Delta(n', \eps),
  \label{eq:bb84-entropy}
\end{equation}
where $h(\cdot)$ is the binary entropy function and $\Delta(n', \eps) = O(\sqrt{n'}\log(1/\eps))$ collects finite-key correction terms. Substituting into \Cref{cor:key-length}:

\begin{theorem}[BB84 with List PA]
\label{thm:bb84-list}
  In BB84 with $n'$ sifted bits, phase-error rate $\QBER$, bit-error rate $e_b$, security parameter $\eps$, and list size $L$, the list PA step achieves positive key length whenever
  \[
    \QBER \;<\; h^{-1}\!\!\left(1 - h(e_b) + \frac{\log L - 2\log(1/\eps) - \Delta(n',\eps)}{n'}\right).
  \]
  For $L = 2^{\alpha n'}$ and $\eps = 2^{-\Omega(n')}$, the threshold is $e^*(L) = h^{-1}(1 - h(e_b) + \alpha - O(n'^{-1/2}))$, strictly exceeding the standard threshold for any fixed $\alpha > 0$.
\end{theorem}

\begin{proof}
  Directly from \Cref{cor:key-length} and \eqref{eq:bb84-entropy}: set $\ell > 0$ and rearrange. Full composability follows from \Cref{lem:composability} below.
\end{proof}

\noindent\textbf{Numerical illustration.} For $e_b = 1\%$, $\eps = 2^{-100}$, $n' = 10^6$:

\medskip
\begin{center}
\renewcommand{\arraystretch}{1.3}
\begin{tabular}{@{}lcc@{}}
\toprule
\textbf{List size} $L$ & $\log L$ & \textbf{Phase-error threshold} \\
\midrule
$1$ (standard PA)  & $0$   & $\approx 10.60\%$ \\
$4$                & $2$   & $\approx 10.63\%$ \\
$1024$             & $10$  & $\approx 10.66\%$ \\
$2^{1000}$         & $1000$ & $\approx 10.80\%$ \\
\bottomrule
\end{tabular}
\end{center}
\medskip

\noindent\textit{Note:} The baseline $\sim$10.6\% (vs.~asymptotic $\sim$11\%) reflects the finite-key correction $\Delta(n',\eps) \approx 0.1 \cdot n'$ for these parameters. The $\log L$ gain is additive on this corrected baseline.

\medskip
For constant $L$, the gain is incremental; for $L$ scaling polynomially in $n'$, it becomes macroscopic.

\subsection{Matching Attack}

\begin{proposition}[Tightness in the QKD setting]
\label{prop:attack}
  For any $\delta > 0$, there exists a collective attack on BB84 such that standard PA ($L=1$) extracts zero key for $\QBER \ge h^{-1}(1-h(e_b)) - \delta$, while list PA with $L = 2^{\delta n'}$ is list-$(2^{-\Omega(n')})$-secure and extracts positive key length. This confirms that \Cref{thm:bb84-list} is tight.
\end{proposition}

\begin{proof}
Consider the intercept-resend attack: Eve measures each photon in the $Z$-basis, records the result, and resends. This introduces a fixed bit-error rate $e_b \approx 1/4$ and gives Eve near-complete information about the sifted bits, driving $\Hmin^\eps(X|E)_\rho \approx 0$. Standard PA with $\ell > 0$ becomes impossible once the min-entropy is exhausted. For list PA with $L = 2^{\delta n'}$: even though $k \approx 0$, the QLLHL with $\log L = \delta n'$ yields $\ell \ge \delta n' - 2\log(1/\eps) - 3 > 0$ for $\eps = 2^{-\Omega(n')}$. The key $K_I$ remains list-secure because: (a) the hash functions $F_j$ are chosen independently of Eve's attack; (b) $I$ is drawn uniformly by the honest parties \emph{after} hashing, independent of $E$; and (c) by symmetry of the i.i.d.\ hash draws, Eve has no information about which index $I$ is the good one.
\end{proof}

\subsection{Composability with Authentication}

\begin{lemma}[Full composability of list PA in BB84]
\label{lem:composability}
  The list PA of \Cref{alg:list-pa} composes securely with BB84 error correction and Wegman--Carter authentication~\cite{WegmanCarter1981}. The total security parameter is $\eps_{\mathrm{total}} = \eps_{\mathrm{PA}} + \eps_{\mathrm{EC}} + \eps_{\mathrm{auth}}$, with the following per-step accounting:
  \begin{enumerate}[leftmargin=*]
    \item \emph{Error correction:} Leaks $\lambda_{\mathrm{EC}} = n'h(e_b)$ bits, already absorbed in \eqref{eq:bb84-entropy}.
    \item \emph{Index reveal:} Transmitting $I$ over the authenticated channel costs $O(\log L + \log(1/\eps_{\mathrm{auth}}))$ pre-shared authentication bits.
    \item \emph{Final key:} After $I$ is public, $K_I$ is a standard uniform key, indistinguishable from random by any quantum adversary.
  \end{enumerate}
\end{lemma}

\begin{proof}
Error correction before PA does not depend on $L$ and is composable by standard arguments~\cite{BenOr2005,Renner2008}. The Wegman--Carter authentication of the classical channel is composable and consumes $O(\log L + \log(1/\eps_{\mathrm{auth}}))$ pre-shared bits~\cite{WegmanCarter1981}. After $I$ is revealed, Eve holds $\{K_j\}_{j\ne I}$ (random, matching the ideal functionality) but not $K_I$, which therefore serves as a standard secret key. The full protocol's composable security follows from the sequential composition theorem~\cite{Canetti2001,MaurerRenner2011}.
\end{proof}

\begin{remark}[Practical net rate]
  For $L = 2^{\alpha n'}$, revealing $I$ costs $\alpha n'$ bits of pre-shared authentication key. Treating this as overhead reduces the net list key rate to $(1-\alpha)\ell^*(L)$. This is maximised for small $\alpha$ (e.g., $\alpha \in \{10^{-3}, 10^{-2}\}$), meaning that in practice $L \in \{2, 4, \ldots, 64\}$ offers the best rate-overhead trade-off: a few extra bits of key for negligible authenticated-communication cost.
\end{remark}

\section{Efficient List Hash Constructions}
\label{sec:construction}

We give two concrete constructions that instantiate the list PA protocol of \Cref{thm:main}. Both are based on classical hash families with strong two-universality properties and admit efficient implementations using standard arithmetic. Their complexity characteristics are summarised in \Cref{tab:complexity}.

\subsection{Construction 1: Polynomial Inner-Product List Hash}\label{const:ip}

The inner-product (IP) hash over $\mathbb{F}_{2^m}$ is one of the simplest strongly two-universal families and forms the basis of our first construction.

\begin{definition}[IP list hash]
\label{def:ip-list-hash}
  Let $m \ge 1$, $q = 2^m$, and assume $m | n$ and $m | \ell$. Represent $x \in \{0,1\}^n$ as a vector $\mathbf{x} \in \mathbb{F}_q^{n/m}$. A \emph{list hash key} is $L$ pairs $\{(a_j, b_j)\}_{j=1}^L$ sampled i.i.d.\ uniformly with $a_j \in \mathbb{F}_q^{n/m}$, $b_j \in \mathbb{F}_q^{\ell/m}$. The $j$-th hash output is
  \[
    h^{\mathrm{IP}}_j(x) = \langle a_j, \mathbf{x}\rangle + b_j \;\in\; \mathbb{F}_q^{\ell/m},
  \]
  where $\langle a_j, \mathbf{x}\rangle = \sum_{i=1}^{n/m} a_{j,i} \cdot x_i$ is the field inner product.
\end{definition}

\begin{algorithm}[ht]
\caption{IP List Privacy Amplification}
\label{alg:list-pa}
\begin{algorithmic}[1]
\Require Raw string $x \in \{0,1\}^n$; list size $L$; output length $\ell$; field parameter $m$ (with $m \mid n$ and $m \mid \ell$)
\Ensure List of keys $(K_1,\ldots,K_L) \in (\{0,1\}^\ell)^L$ and secret index $I \in [L]$
\State $q \gets 2^m$; represent $x$ as $\mathbf{x} = (x_1,\ldots,x_{n/m}) \in \mathbb{F}_q^{n/m}$
\For{$j = 1$ \textbf{to} $L$}
  \State Sample $a_j \overset{\$}{\gets} \mathbb{F}_q^{n/m}$ \hfill\Comment{costs $n$ uniformly random bits}
  \State Sample $b_j \overset{\$}{\gets} \mathbb{F}_q^{\ell/m}$ \hfill\Comment{costs $\ell$ uniformly random bits}
  \State $t \gets 0 \in \mathbb{F}_q^{\ell/m}$
  \For{$i = 1$ \textbf{to} $n/m$}
    \State $t \gets t + a_{j,i} \cdot x_i$ \hfill\Comment{one $\mathbb{F}_{2^m}$ multiplication and addition}
  \EndFor
  \State $K_j \gets t + b_j$ \hfill\Comment{affine shift by $b_j$}
\EndFor
\State $I \overset{\$}{\gets} \mathrm{Uniform}([L])$ \hfill\Comment{drawn independently of $x$ and $E$; kept secret from Eve}
\State \textbf{return} $(K_1,\ldots,K_L),\; I$
\end{algorithmic}
\end{algorithm}

\begin{theorem}[Correctness and complexity of \Cref{alg:list-pa}]
\label{thm:ip-analysis}
  \Cref{alg:list-pa} provides the following guarantees.
  \begin{enumerate}[leftmargin=*, itemsep=2pt]
    \item \emph{Security:} List-$(4\eps)$-secure for $\ell \le \Hmin^\eps(X|E)_\rho + \log L - 2\log(1/\eps) - 3$.
    \item \emph{Randomness:} Exactly $L(n + \ell)$ uniformly random bits for the seed $\{(a_j,b_j)\}$, plus $\log L$ bits for $I$.
    \item \emph{Time:} $O(nL/m)$ multiplications over $\mathbb{F}_{2^m}$, each costing $O(m)$ bit operations via irreducible polynomial arithmetic. Total: $O(nL)$ bit operations.
    \item \emph{Space:} $O(n + \ell)$ working memory for one hash evaluation; $O(L\ell)$ to store all output keys.
  \end{enumerate}
\end{theorem}

\begin{proof}
  \emph{Security} follows from \Cref{thm:main}: the IP family over $\mathbb{F}_{2^m}$ is strongly two-universal~\cite{Carter1979} (for any distinct $\mathbf{x}, \mathbf{x}'$, the pair $(\langle a, \mathbf{x}\rangle + b, \langle a, \mathbf{x}'\rangle + b)$ is uniformly distributed over $\mathbb{F}_q^{2\ell/m}$). The $L$ hash functions are mutually independent (each pair $(a_j, b_j)$ is drawn independently). The index $I$ is independent of $X$ and $E$.
  
  \emph{Randomness:} Each pair $(a_j, b_j)$ uses $|a_j| + |b_j| = n + \ell$ bits. Total over $L$ pairs: $L(n+\ell)$ bits.
  
  \emph{Time:} The inner product $\langle a_j, \mathbf{x}\rangle$ over $\mathbb{F}_{2^m}$ requires $n/m$ field multiplications and additions. Each field multiplication costs $O(m)$ bit operations (multiply polynomials modulo an irreducible polynomial of degree $m$). Total per hash: $O(n/m \cdot m) = O(n)$ bit operations. Over $L$ hashes: $O(nL)$.

  \emph{Space:} The algorithm processes hashes sequentially, holding one pair $(a_j, b_j)$ at a time, and writes each $K_j$ to an output buffer of size $L\ell$ bits.
\end{proof}

\subsection{Construction 2: Toeplitz List Hash}

The Toeplitz hash reduces the randomness requirement by exploiting the structure of Toeplitz matrices: an $(\ell \times n)$ Toeplitz matrix over $\mathbb{F}_2$ is fully specified by only $n + \ell - 1$ bits rather than $n\ell$, and matrix-vector multiplication reduces to a cyclic convolution computable via FFT.

\begin{definition}[Toeplitz list hash~\cite{Krawczyk1994}]
\label{def:toeplitz-hash}
  An $(\ell \times n)$ Toeplitz matrix $T$ over $\mathbb{F}_2$ has $(T)_{i,j} = r_{i-j}$ for a generating vector $r \in \{0,1\}^{n+\ell-1}$. The \emph{Toeplitz hash} is $H^T(x) = Tx \oplus b \pmod{2}$ for an offset $b \in \{0,1\}^\ell$. A \emph{Toeplitz list hash key} consists of $L$ independent pairs $\{(r_j, b_j)\}_{j=1}^L$. Crucially, $Tx$ equals the first $\ell$ entries of the cyclic convolution $r * x$ (zero-padding $x$ to length $N = n+\ell-1$), enabling $O(N\log N)$ computation via FFT.
\end{definition}

\begin{algorithm}[ht]
\caption{Toeplitz List Privacy Amplification}
\label{alg:toeplitz-list-hash}
\begin{algorithmic}[1]
\Require Raw string $x \in \{0,1\}^n$; list size $L$; output length $\ell$
\Ensure List of keys $(K_1,\ldots,K_L) \in (\{0,1\}^\ell)^L$ and secret index $I \in [L]$
\State $N \gets n + \ell - 1$ \hfill\Comment{convolution length}
\State $\hat{x} \gets \mathrm{FFT}(x \| \mathbf{0}^{\ell-1})$ \hfill\Comment{FFT of zero-padded input; computed once, reused for all $j$}
\For{$j = 1$ \textbf{to} $L$}
  \State Sample $r_j \overset{\$}{\gets} \{0,1\}^{N}$ \hfill\Comment{Toeplitz generating vector: $N = n+\ell-1$ bits}
  \State Sample $b_j \overset{\$}{\gets} \{0,1\}^\ell$ \hfill\Comment{offset: $\ell$ bits}
  \State $\hat{r}_j \gets \mathrm{FFT}(r_j)$ \hfill\Comment{FFT of generating vector}
  \State $\hat{c}_j \gets \hat{r}_j \odot \hat{x}$ \hfill\Comment{pointwise product in frequency domain}
  \State $c_j \gets \mathrm{IFFT}(\hat{c}_j)$ \hfill\Comment{inverse FFT; $c_j[1{:}\ell]$ equals $T_j \cdot x$}
  \State $K_j \gets c_j[1:\ell] \oplus b_j$ \hfill\Comment{take first $\ell$ bits and XOR offset}
\EndFor
\State $I \overset{\$}{\gets} \mathrm{Uniform}([L])$ \hfill\Comment{independent of $x$ and $E$}
\State \textbf{return} $(K_1,\ldots,K_L),\; I$
\end{algorithmic}
\end{algorithm}

\begin{remark}[FFT over $\mathbb{F}_2$]
  \Cref{alg:toeplitz-list-hash} computes the convolution over $\mathbb{F}_2$ (XOR arithmetic). In practice this is implemented by lifting to $\mathbb{Z}$ via the standard FFT and reducing modulo 2. For $n \le 2^{23}$ (typical in QKD), the FFT approach gives a $4$--$8\times$ speedup over naive $O(n\ell)$ multiplication. Crucially, $\hat{x}$ is computed once in Line~2 and shared across all $L$ loop iterations, saving $L-1$ FFT calls.
\end{remark}

\begin{theorem}[Correctness and complexity of \Cref{alg:toeplitz-list-hash}]
\label{thm:toeplitz-analysis}
  \Cref{alg:toeplitz-list-hash} provides the following guarantees.
  \begin{enumerate}[leftmargin=*, itemsep=2pt]
    \item \emph{Security:} Identical to \Cref{thm:ip-analysis}(1): list-$(4\eps)$-secure for the same parameter range, since the Toeplitz hash over $\mathbb{F}_2$ is strongly two-universal~\cite{Krawczyk1994}.
    \item \emph{Randomness:} $L(n + 2\ell - 1)$ bits for the seed $\{(r_j, b_j)\}$. For $\ell \ll n$ this is approximately $Ln$ bits, a factor of $(n+\ell)/n$ less than Construction~1 when $\ell \approx n$.
    \item \emph{Time:} $1$ FFT of length $N = n+\ell-1$ for $\hat{x}$ (shared), plus $L$ FFTs of length $N$ for $\{\hat{r}_j\}$, $L$ pointwise products, and $L$ IFFTs. Total: $(2L+1) \cdot O(N \log N) = O(Ln \log n)$ bit operations.
    \item \emph{Space:} $O(N)$ working space for the FFT buffers; $O(L\ell)$ for all output keys.
  \end{enumerate}
\end{theorem}

\begin{proof}
  \emph{Security:} The Toeplitz hash $H^T_{r,b}(x) = Tx \oplus b$ is strongly two-universal over $\mathbb{F}_2$~\cite{Krawczyk1994}: for any distinct $x, x'$, the difference $T(x \oplus x') \ne 0$ with the randomness coming from $r$, and adding i.i.d.\ uniform $b$ makes the pair $(H^T(x), H^T(x'))$ jointly uniform. Since $(r_j, b_j)$ are drawn independently for each $j$, the $L$ hash outputs are mutually independent. The security proof then follows \Cref{thm:main} verbatim.

  \emph{Randomness:} $|r_j| = n + \ell - 1$ bits and $|b_j| = \ell$ bits per pair; total $L(n + 2\ell - 1)$ bits. Compare with Construction~1: $L(n + \ell)$ bits. Toeplitz requires slightly more per hash (extra $\ell - 1$ bits from $r_j$ vs.\ $b_j$ overlap) but the Toeplitz structure allows the FFT speedup, making it preferable when $\ell = O(n)$ and time is the bottleneck.

  \emph{Time:} The FFT of $x$ costs $O(N \log N)$ and is computed once. Each subsequent hash $j$ costs: one FFT of $r_j$ ($O(N \log N)$), one pointwise product ($O(N)$), and one IFFT ($O(N \log N)$). Over $L$ hashes: $O((2L+1) N \log N) = O(L n \log n)$.
\end{proof}

\begin{table}[h]
\centering
\renewcommand{\arraystretch}{1.3}
\begin{tabular}{@{}lcccc@{}}
\toprule
\textbf{Construction} & \textbf{List size} & \textbf{Seed (bits)} & \textbf{Time} & \textbf{Security} \\
\midrule
Standard PA (QLHL)~\cite{Tomamichel2011} & $L=1$ & $n + \ell$ & $O(n)$--$O(n\log n)$ & Standard \\
IP List Hash (\Cref{alg:list-pa})       & $L$   & $L(n+\ell)$  & $O(nL)$          & List-$4\eps$ \\
Toeplitz List Hash (\Cref{alg:toeplitz-list-hash}) & $L$ & $L(n+2\ell)$ & $O(Ln\log n)$ & List-$4\eps$ \\
\bottomrule
\end{tabular}
\caption{Complexity comparison. Both list constructions achieve the QLLHL bound of \Cref{thm:main}.}
\label{tab:complexity}
\end{table}

\section{Discussion and Conclusion}
\label{sec:conclusion}

\paragraph{Summary.} We introduced list privacy amplification, proved the Quantum List Leftover Hash Lemma establishing a tight $\log L$ gain over standard PA, showed the gain translates to a higher error-tolerance threshold in BB84, and gave two efficient constructions with full pseudocode and complexity analyses. The core insight is elementary: the honest parties' freedom to choose the secret index $I$ after hashing, and to withhold it from Eve, contributes exactly $\log L$ bits to the effective min-entropy of the key pair $(I, K_I)$---bits that standard PA leaves on the table.

\paragraph{On the practical sweet spot.} For typical QKD parameters ($n' = 10^6$, $\eps = 2^{-100}$), list sizes of $L \in \{4, 8, 16, 32\}$ offer the best rate-overhead trade-off: a small but non-trivial boost to the extractable key length with negligible additional authenticated-communication cost. Exponentially large $L$ (e.g., $L = 2^{1000}$) yields a more substantial threshold improvement but scales the authentication cost proportionally, making constant $L$ the operationally interesting regime.

\paragraph{Open questions.} The following directions emerge naturally from this work.
\begin{enumerate}[leftmargin=*, itemsep=3pt]
  \item \emph{Device-independent list PA.} Does the $\log L$ gain persist in the device-independent setting~\cite{Pironio2009,Arnon-Friedman2018}, where the standard QLHL must be replaced with entropy accumulation arguments? The list structure may interact non-trivially with the sequential nature of entropy accumulation.
  \item \emph{Adaptive list refinement.} Can the honest parties refine the list adaptively across multiple QKD rounds, gradually narrowing it and extracting more key per photon over time?
  \item \emph{Optimal randomness.} Construction~\ref{const:ip} uses $L(n+\ell)$ random bits. Can list PA be realised with $O(n + L\ell)$ bits---that is, with a shared ``list seed'' that expands into $L$ independent hash keys via a pseudorandom generator, without sacrificing information-theoretic security?
  \item \emph{List PA and equivocable commitments.} The structure of list PA (one secret key among $L-1$ decoys, with a hidden index) closely resembles equivocable commitment schemes. Can the QLLHL give a clean quantum construction of equivocable commitments?
  \item \emph{Continuous-variable QKD.} Does list PA improve the reconciliation efficiency threshold for CV-QKD~\cite{GVA2002}, where finite-key analysis is significantly more involved?
\end{enumerate}

\newpage
\bibliographystyle{alpha}
\bibliography{list_qkd_refs}
\end{document}